\documentclass[letterpaper, 10 pt, conference]{ieeeconf}  
\usepackage[font=footnotesize, compatibility=false]{caption}
\usepackage{xurl}
\usepackage{breakurl}
\usepackage{hyperref}



\usepackage{amssymb}
\usepackage{amsmath}
\usepackage{amsthm}
\newtheorem{theorem}{Theorem}
\newtheorem{lemma}{Lemma}
\newtheorem{remark}{Remark}

\usepackage{tikz}
\usepackage{textcomp}
\usepackage{hyperref}
\usepackage{lipsum}

\newcommand\copyrighttext{%
  \footnotesize \textcopyright 2026 IEEE. Personal use of this material is permitted.
  Permission from IEEE must be obtained for all other uses, in any current or future 
  media, including reprinting/republishing this material for advertising or promotional 
  purposes, creating new collective works, for resale or redistribution to servers or 
  lists, or reuse of any copyrighted component of this work in other works.} 
\newcommand\copyrightnotice{%
\begin{tikzpicture}[remember picture,overlay]
\node[anchor=south,yshift=10pt] at (current page.south) {\fbox{\parbox{\dimexpr\textwidth-\fboxsep-\fboxrule\relax}{\copyrighttext}}};
\end{tikzpicture}%
}

\usepackage{algorithmic}
\usepackage[ruled,vlined,linesnumbered]{algorithm2e}

\usepackage{graphicx}
\usepackage{xcolor}

\IEEEoverridecommandlockouts                              

\title{\LARGE \bf
Active Calibration of Reachable Sets Using Approximate Pick-to-Learn
}

\author{Sampada Deglurkar$^{1}$, Ebonye Smith$^{1}$, Jingqi Li$^{2}$, and Claire J. Tomlin$^{1}$
\thanks{This work was supported by NSF Safe Learning Enabled Systems, 
the DARPA Assured Autonomy, ANSR, and TIAMAT programs, and
the NASA ULI on Safe Aviation Autonomy.  
E.S. was supported by NSF GRFP, and J.L. by an Oden
Institute Fellowship.}
\thanks{$^{1}$Sampada Deglurkar, Ebonye Smith, and Claire J. Tomlin are with the Department of Electrical Engineering and Computer Sciences, University of California Berkeley, Berkeley, CA 94720, USA, tel:510-643-6610,  {\tt\small sampada\_deglurkar@berkeley.edu, ebonyesmith@berkeley.edu, tomlin@berkeley.edu}}%
\thanks{$^{2}$Jingqi Li is  with the Oden Institute of Computational Engineering and Sciences, University of Texas at Austin, Austin, TX 78712, USA, tel: 765-337-0678
        {\tt\small jingqi.li@austin.utexas.edu}}%
}

\begin{document}

\maketitle
\thispagestyle{empty}
\pagestyle{empty}

\copyrightnotice

\begin{abstract}

Reachability computations that rely on approximate or learned models require calibration in order to uphold confidence about their guarantees.
Calibration generally involves sampling scenarios inside the reachable set. However, producing reasonable probabilistic guarantees may require many samples, which can be costly.
To remedy this, we propose that calibration of reachable sets be performed using active learning strategies.
In order to produce a probabilistic guarantee on the active learning, we adapt the Pick-to-Learn algorithm, which produces generalization bounds for standard supervised learning, to the active learning setting.
Our method, Approximate Pick-to-Learn, treats the process of choosing data samples as maximizing an approximate error function.
Conformal prediction is used to ensure that the approximate error is close to the true model error.
We demonstrate our technique for a simulated drone racing example in which learning is used to provide an initial guess of the reachable tube.
Our method requires fewer samples to calibrate the model and provides more accurate sets than the baselines while also producing a novel generalization bound.
\end{abstract}

\section{INTRODUCTION}

In many safety-critical applications, reachability analysis \cite{bansal2017hamiltonjacobi} is a valuable tool for providing safety guarantees for dynamical systems.
Traditionally, reachability methods are computationally intensive, are performed offline, and require pre-specified models of dynamics, objectives, and constraints.
Some of these limitations have been addressed by the use of learning-based algorithms to solve for reachable sets \cite{bansal2021deepreach,hsu2021safety,li2025certifiable}.
However, not only can learning-enabled techniques be error-prone, but also system modeling assumptions may not match reality.
For example, \cite{lin2023generating} provides examples of how incomplete training can cause a learned reachable set to misclassify unsafe points as safe near its boundary.
In the absence of hard safety guarantees, one approach from prior literature is to turn to \textit{probabilistic} guarantees, and to \textit{calibrate} a reachable set using samples~\cite{lin2023generating, lin2024verification}.  

\begin{figure}[!t]
\centering
\includegraphics[width=3.35in]{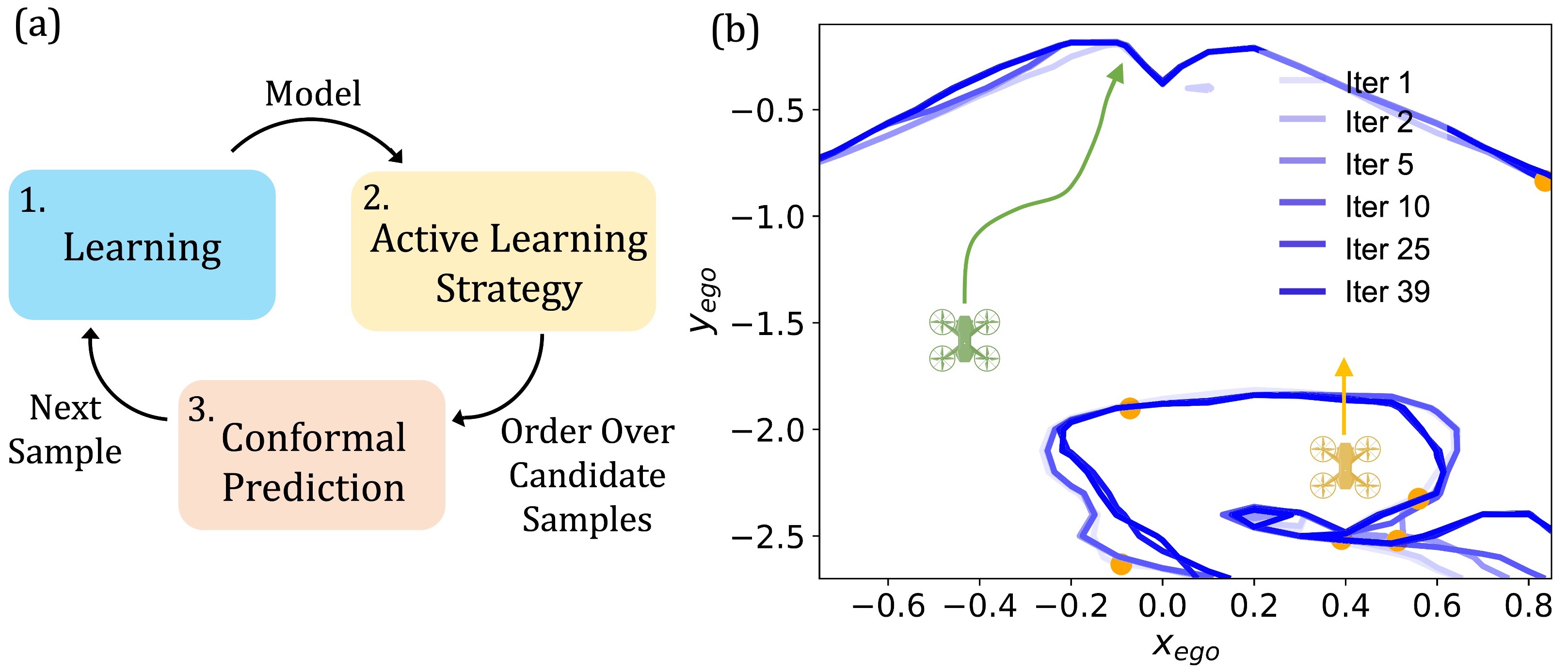}\vspace{-0.2em}
\caption{(a) Our framework provides a probabilistic generalization bound for the \textit{iterative} process of learning and actively sampling the next data point.
The novelty of our method lies in the realization that if we conformally calibrate our active sampling strategy so that it approximates maximizing the true model error, we can achieve a desired probabilistic guarantee.
(b) We demonstrate our method on a drone racing example in which the ego drone (green) overtakes another drone (yellow).
We display the reach-avoid set learned by our method, with the lighter colors indicative of the set at previous iterations and the darker colors for later iterations.
The orange points are samples taken by our calibrated active learning strategy.
}
\label{fig:front_fig}
\end{figure}

However, producing a strong probabilistic guarantee may require taking a large amount of potentially costly samples.
Also, as the calibration procedure \textit{adapts} the reachable set to the environment, the set updates in response to each sample that is taken.
This \textit{iterative} learning and sampling breaks assumptions made by many prior works that the samples be exchangeable or IID.
Finally, there is the question of how to define calibration in the first place.
Understanding how to update a reachable set is a learning problem in itself, which prior works often simplify by making the calibrated reachable set a simple function of the original set.

We address these issues by framing calibration as an active learning problem, through which we can choose only the necessary samples.
At the same time, we provide strong guarantees and high-accuracy reachable sets.
For this, our key insight lies in leveraging the Pick-to-Learn algorithm, introduced by Paccagnan \textit{et al.} \cite{paccagnan2023picktolearn}.
This is a meta-algorithm that turns any supervised learning procedure into a compression scheme, which yields a generalization bound.
We adapt this algorithm for active learning, resulting in our \textit{Approximate Pick-to-Learn} method.
We overview our method in Fig. \ref{fig:front_fig}.
In summary, our contributions are:

\begin{itemize}
    \item We frame reachable set calibration as an active learning problem and provide a novel method to produce a probabilistic guarantee.
    \item We introduce an extension to the Pick-to-Learn algorithm for active learning in the form of Approximate Pick-to-Learn and theoretically establish its probabilistic guarantees.
    \item We demonstrate our methodology in simulation and show how it balances tradeoffs between sample complexity, guarantee strength, and reachable set accuracy better than the baselines.
\end{itemize}

In this paper, we first give an overview of related work in Section \ref{sec:related_work}. We then describe the problem setup in Section \ref{sec:problem_setup}, explain our method in Section \ref{sec:approx_picktolearn}, show our simulation experiments in Section \ref{sec:simulation_exp}, and conclude in Section \ref{sec:conclusion}.

\section{RELATED WORK}\label{sec:related_work}

\subsection{Calibrating Learned Reachability Computations} \label{sec:related_work_A}

The works \cite{lin2023generating} and \cite{lin2024verification}, which serve as our baselines, study the calibration of learned reachable sets using sampled trajectories.
The iterative method \cite{lin2023generating} evaluates ground-truth reachability values at IID samples and adjusts the level set of the value function to obtain zero safety violations, while the robust method \cite{lin2024verification} instead counts violations and incorporates them into a relaxed probabilistic guarantee. Both rely on scenario optimization \cite{campi2018introduction}, trading off the strength of the guarantee with the sample complexity. Our work builds on these ideas with a more sample-efficient, active learning-based calibration approach.

There are also prior works on adapting reachable sets online, such as in response to new sensed information \cite{BajcsyBBTT19}, with respect to changing dynamics or environmental parameters \cite{NakamuraBansal2022OnlineUpdate, BorquezNakamuraBansal2023PCRS, JeongGongBansalHerbert2024ParamFaSTrack}, or in accordance with human feedback about constraints \cite{SantosLPBB25, AgrawalSeoNakamuraTianBajcsy2025AnySafe}.
However, these works constrain the amount of adaptation that can happen and do not have probabilistic guarantees on the final reachable set.

\subsection{Guarantees for Active and Online Learning Settings}

Many works in the multi-armed bandits or Bayesian optimization literature make regret-related theoretical statements \cite{TakenoInatsuKarasuyamaTakeuchi2025PSB_EI, ChengAstudilloDesautelsYue2024PSBAX, ZhangChen2025DirectRegretBO,BubeckCesaBianchi2012Bandits}, which describe the optimality of the active learning process but not how the final model performs on unseen data.
They may otherwise assume that the learned model belongs to a particular hypothesis class \cite{MasonCamilleriMukherjeeJamiesonNowakJain2023LevelSet, FlynnReebKandemirPeters2023PACBayesBandits}.
For example, some Bayesian optimization works produce probabilistic statements on the accuracy of level set estimates, but with the assumption that the underlying function is drawn from a Gaussian process (GP) \cite{Ishibashi2024EpsDeltaLevelSet, InatsuIwazakiTakeuchi2024DRLSE}.
In contrast, we work with arbitrary and potentially black-box models of the reachability value function.

Another line of related work applies conformal prediction \cite{angelopoulos2023conformal} to produce probabilistic safety guarantees in dynamic settings.
For example, \cite{muthali2023multiagent, DixitLindemannWeiCleavelandPappasBurdick2023ACP} use a form of Adaptive Conformal Inference \cite{gibbs2022conformal}, whose guarantee is similar to a regret bound.
Our method uses conformal prediction, but to calibrate the estimate of where it is important to sample next rather than to calibrate the reachable set itself.

\section{Problem Definition}\label{sec:problem_setup}

Let $x \in \mathcal{X}$ denote a state in the bounded state space, and let $x_{t+1} = f(x_t, u_t)$ describe the system dynamics, where $x_t\in\mathcal{X}$ and $u_t \in \mathcal{U}$ denote the state and control at time step $t$, respectively.
For an arbitrary policy $\pi: \mathcal{X} \to U$, the value function induced by the policy is 
\begin{equation}\label{eqn:V_defn}
    V_{\pi}(x, t)\hspace{-0.1cm} :=\hspace{-0.1cm}  \sup_{t'=0,...,t}\hspace{-0.1cm}\min\{\gamma^{t'}r(\xi_x^{\pi}(t')), \min_{\tau=0,...,t'}\gamma^{\tau}\hspace{-0.1cm}c(\xi_x^{\pi}(\tau))\}
\end{equation}
for discount factor $\gamma \in (0,1)$, where $\xi_x^{\pi}: [0, t] \to \mathcal{X}$ represents the trajectory under $\pi$ starting at $x$.
Here, $r: \mathcal{X} \to \mathbb{R}$ and $c: \mathcal{X} \to \mathbb{R}$ are functions such that $r(x) > 0$ indicates that $x$ lies in the target set, and $c(x) > 0$ indicates that $x$ satisfies constraints. 
The super-zero level set of $V_{\pi}(x,t)$ recovers the \emph{reach-avoid set}, which consists of all states from which the policy guides the system to the target set within $t$ stages while satisfying the constraints.

In our experiments, we obtain a policy $\tilde{\pi}(x)$ and associated value function $\tilde{V}(x,t)$ using the RL-based method in \cite{li2025certifiable}.
Since the learned value function may be subject to error, we use it to initialize our active learning calibration process.
Let the \textit{hypothesis} $h: \mathcal{X} \to \mathbb{R}$, with initial hypothesis $h_0$, refer to an estimate of $V_{\tilde{\pi}}(x, T)$ for time horizon $T$.
During the active learning procedure, we iteratively choose $x$'s and observe $V_{\tilde{\pi}}(x,T)$ by rolling out $\tilde{\pi}$ and applying (\ref{eqn:V_defn}). 
This provides samples $z := (x, V_{\tilde{\pi}}(x,T))$.
Let $z_x$ denote the state portion of $z$.
We also assume access to a learning algorithm $L$, which maps a list of samples to a hypothesis, and an error function $e_h(z)$.
In our Approximate Pick-to-Learn setting, this error is binary-valued: $e_h(z) = \mathbf{1}\{h(x) \times V_{\tilde{\pi}}(x, T) \leq 0 \; \land  \; h(x) \neq V_{\tilde{\pi}}(x, T) \}$,
so it is $0$ when $h$ correctly predicts the sign of the value and $1$ otherwise.
Our aim is to improve our hypothesis $h$ while also using the collected data to produce a probabilistic guarantee.

\section{APPROXIMATE PICK-TO-LEARN}\label{sec:approx_picktolearn}

In this section, we describe our adaptation of the Pick-to-Learn algorithm for active learning.
Section \ref{sec:picktolearn_overview} provides an overview of Pick-to-Learn.
We then describe in Section \ref{sec:adapting_picktolearn} the key ingredients that should be adapted for our algorithm: Pick-to-Learn's dataset and hypothesis-dependent total order.
Finally, Section \ref{sec:calibration} describes how we use conditional conformal prediction to produce a calibrated \textit{approximate} total order and obtain our final guarantee \footnote{Code for this project can be found at \url{https://github.com/sdeglurkar/bayes_opt_calibration}.}.

\subsection{The Pick-to-Learn Algorithm} \label{sec:picktolearn_overview}

In the standard supervised learning setting, the Pick-to-Learn algorithm \cite{paccagnan2023picktolearn} produces a generalization bound by turning any learning algorithm into a \textit{compression scheme}. 
For a dataset $D:= \{z_1,...,z_{n_D}\}$ with IID $z_i$'s in $\mathcal{Z}$, the method aims to iteratively build the compressed set $Q \subseteq D$.
In our case, $\mathcal{Z} := \mathcal{X} \times \mathbb{R}$ so each $z$ is a combination of state and value label.
At each iteration, the data point in $D \setminus Q$ on which the current hypothesis performs the worst is added to $Q$.
This is the point that maximizes a \textit{hypothesis-dependent total order} $\leq_h$ defined over $\mathcal{Z}$: $z_i \leq_h z_j \iff e_h(z_i) \leq e_h(z_j), \; \forall z_i, z_j \in \mathcal{Z}$.
Then, $L$ is applied to obtain a new $h$.
The algorithm terminates when $e_h(z) \leq \omega \; \forall z \in D \setminus Q$ for the threshold $\omega$ and returns $Q$ and the final hypothesis $h_f$.

\subsection{Adapting Pick-to-Learn for Active Learning} \label{sec:adapting_picktolearn}

As Pick-to-Learn is designed for offline supervised learning, several aspects need to be changed to apply the technique to active learning.
Firstly, in active learning, a \textit{labeled} dataset $D$ does not exist.
Because of this, we also do not have access to $e_h(z)$ for all $z \in D$, and so secondly, we do not know the hypothesis-dependent total order.

Our response to these dilemmas is to redefine $D$ as an \textit{unlabeled} dataset of $n_D$ IID samples of $x$ from a uniform distribution over the state space; $D := \{x_1,...x_{n_D}\}$.
Though now $D \subseteq \mathcal{X}$, we will allow $Q$ to have labels so that $Q \subseteq \mathcal{Z}$.
In this way, we think of $D$ as a representation of the world in which we are performing active learning; we can choose which elements of $D$ we would like to learn on.
In the Pick-to-Learn algorithm, this choice is precisely adversarial in that it is the point for which $e_h(z)$ is maximized.
However, in the absence of knowing $e_h(z)$ exactly, one can imagine designing an active learning strategy that chooses points that reduce uncertainty, optimize an objective, or follow a heuristic \cite{garnett2023bayesopt}.
Generally speaking, an active learning strategy can be modeled as maximizing some function $a_{h}(x)$.

\begin{algorithm}[t]
\caption{$\mathcal{A}_{C}(D)$ - Approximate Pick-to-Learn for Calibrating Reachable Sets}
\label{alg:approx_picktolearn}
\begin{algorithmic}[1]
\REQUIRE $Q = \emptyset$, $h = h_0$, $C = \{z_1^C,...,z_{n_C}^C\}$, 
$\lambda_{h_0,C} = Quantile_{\lceil(1-\alpha)(n_C+1)\rceil/n_C}$
$\Bigl\{\frac{|e_{h_0}(z_1^C) - a_{h_0}(z_{1,x}^C)|}{\mu_{h_0}(z_{1,x}^C)},...,\frac{|e_{h_0}(z_{n_C}^C) - a_{h_0}(z_{{n_C},x}^C)|}{\mu_{h_0}(z_{{n_C},x}^C)}\Bigr\}$, $\hat{e}_{h_0,C} = a_{h_0} + \lambda_{h_0,C}\mu_{h_0}$, $\bar{x} = \arg \max _{x \in D} \hat{e}_{h_0,C}(x)$

\WHILE{$\hat{e}_{h, C}(\bar{x}) \geq \omega$}
    \STATE \text{Obtain} $V_{\tilde{\pi}}(\bar{x}, T)$ \text{and define} $\bar{z} := (\bar{x}, V_{\tilde{\pi}}(\bar{x}, T))$
    \STATE $Q \leftarrow Q \cup \{\bar{z}\}$ 
    \STATE $h \leftarrow L(Q)$
    \RETURN $h, Q$, \textbf{ if } $D \setminus Q = \emptyset$
    \STATE $\lambda_{h,C} \leftarrow Quantile_{\lceil(1-\alpha)(n_C+1)\rceil/n_C}$
     $\Bigl\{\frac{|e_h(z_1^C) - a_{h}(z_{1,x}^C)|}{\mu_{h}(z_{1,x}^C)},...,\frac{|e_h(z_{n_C}^C) - a_{h}(z_{{n_C},x}^C)|}{\mu_{h}(z_{{n_C},x}^C)}\Bigr\}$
    \STATE $\hat{e}_{h,C} \leftarrow a_{h} + \lambda_{h,C} \mu_{h}$ 
    \STATE $\bar{x} \leftarrow \arg \max _{x \in D \setminus Q} \hat{e}_{h,C}(x)$  
\ENDWHILE
\RETURN $h, Q$
\end{algorithmic}
\end{algorithm}

\subsection{Calibrating the Active Learning Strategy with Conditional Conformal Prediction} \label{sec:calibration}

In Pick-to-Learn, an upper bound is produced on the \textit{risk} of the hypothesis, which is defined as $P_z(e_h(z) \geq \omega | D)$.
However, if the termination condition of our algorithm were $a_{h}(x) \leq \omega \; \forall x \in D$, we would only obtain an upper bound on $P_x(a_{h}(x) \geq \omega | D)$.
To address this issue, we leverage conformal prediction with calibration set $C$ to produce a value $\hat{e}_{h,C}(x)$, serving as an adjusted version of $a_{h}(x)$, so that with high probability, both $e_h(z) \geq \omega$ and $\hat{e}_{h, C}(z_x) \geq \omega$ are true for any $z_x$ and $h$.

We detail our method in Algorithm \ref{alg:approx_picktolearn}.
Following one technique given in \cite{angelopoulos2023conformal}, we construct an interval around $a_{h}(x)$.
Given a heuristic value $\mu_{h}(x)$ that serves as an estimate of the distance between $a_{h}$ and $e_h$, conditional conformal prediction calculates $\lambda_{h, C} \in \mathbb{R}$ so that:

\begin{equation} \label{eqn:conf_prediction}
\begin{split}
    P_C[1-\alpha-\epsilon_{\alpha} \leq P_z[e_h(z) \in [a_{h}(z_x) - \lambda_{h,C} \mu_{h}(z_x), \\ a_{h, \eta}(z_x)\hspace{-0.05cm} +\hspace{-0.05cm} \lambda_{h,C} \mu_{h,}(z_x) \;\hspace{-0.05cm} |\hspace{-0.05cm} \; C, D] \hspace{-0.1cm}\leq 1\hspace{-0.1cm}-\alpha+\epsilon_{\alpha} \; \hspace{-0.05cm}|\hspace{-0.05cm} \; D] \hspace{-0.05cm}\geq\hspace{-0.05cm} 1-\beta
\end{split}
\end{equation}
for calibration set $C := \{z_1^C,...,z_{n_C}^C\}$ containing IID state samples separate from $D$.
In essence, $C$ is a set of state samples for which we have value labels and therefore know $e_h(z)$ for calibration purposes.
$\lambda_{h,C}$ is the $\lceil(1-\alpha)(n_C+1)\rceil/n_C$ quantile of $\{s(z_1^C),...,s(z_{n_C}^C)\}$ for score function $s(z) := |e_h(z) - a_{h}(z_x)|/\mu_{h}(z_x)$ (Line 6 in Algorithm \ref{alg:approx_picktolearn}).
From \cite{angelopoulos2023conformal}, the inner probability in Equation (\ref{eqn:conf_prediction}) follows a Beta distribution.
For some constant $\beta$, the size of $C$ is calculated such that $1-\beta$ probability mass of the distribution remains within $1-\alpha-\epsilon_{\alpha}$ and $1-\alpha+\epsilon_{\alpha}$ for tolerance $\epsilon_{\alpha}$.

Now, we write $\hat{e}_{h, C}(x) := a_{h}(x) + \lambda_{h,C} \mu_{h}(x)$ (Line 7).
This produces a total order $\hat{\leq}_{h}$ that \textit{approximates} $\leq_h$ over $z_x$: $x_i \; \hat{\leq}_{h} \; x_j \iff \hat{e}_{h, C}(x_i) \leq \hat{e}_{h, C}(x_j), \; \forall x_i, x_j \in \mathcal{X}$.
We impose that $\mu_{h}(x)$ must never be $0$.
The Pick-to-Learn algorithm can now be performed as is except by maximizing $\hat{e}_{h,C}(x)$ at every step (Line 8) and with the termination condition $\hat{e}_{h,C}(x) \leq \omega$ for all states in $D \setminus Q$ (Line 1). 
We present our final guarantee with the following theorem:

\begin{theorem}\label{thm:guarantee}
    Let $\mathcal{A}_{C}$ be Algorithm \ref{alg:approx_picktolearn}, such that $\mathcal{A}_{C}(D) = (h_f, Q)$. Pick $\alpha, \epsilon_{\alpha}, \beta,\delta \in (0, 1)$. 
Then,
\begin{equation*}
    P_{C,D}[P_z(e_{h_f}(z) \geq \omega \; | \; C,D) \leq \bar{\epsilon}(|Q|, \delta) + \alpha + \epsilon_{\alpha}] \geq 1-\beta-\delta,
\end{equation*}
where $\bar{\epsilon}(n,\delta)$ for $n = |Q|$ is the unique solution of $\Psi_{n,\delta}(\epsilon) = \frac{\delta}{2n_D} \sum_{m=n}^{n_D-1} \frac{\binom{m}{n}}{\binom{n_D}{n}} (1-\epsilon)^{-(n_D-m)} + \frac{\delta}{6n_D}\sum_{m=n_D+1}^{4n_D} \frac{\binom{m}{n}}{\binom{n_D}{n}} (1-\epsilon)^{m-n_D} = 1$ in the interval $[n/n_D, 1]$ and $\bar{\epsilon}(n_D,\delta)~=~1$.
\end{theorem}
\begin{proof}
We provide a proof sketch here and defer the detailed proof to the Appendix.
We use Theorem 4 from Campi and Garatti (2023) \cite{campi2023compression} to probabilistically bound a quantity called the \textit{probability of change of compression}.
The theorem can be used to bound the probability that the error of the hypothesis is high ($e_{h_f}(z) \geq \omega$) if the event of the error being high is related to the event of the change of compression.
Thus, we take the following general steps: we first establish that with high probability over $C$, it is unlikely that $e_h(z) \geq \omega$ but $\hat{e}_{h,C}(z_x) \leq \omega$.
We then relate the event $e_h(z) \leq \omega \; \lor \; \hat{e}_{h,C}(z_x) \geq \omega$ to the fact that $\mathcal{A}_{C}$ would produce a different output $Q$ were it given that additional $z$ (change of compression).
Then, we can combine various probabilistic statements to arrive at our final statement.
\end{proof}

Theorem~\ref{thm:guarantee} suggests that, with high probability, the learned value function has low error on unseen states even though the samples used for verification came from an adaptive process.

\begin{remark}
Although we have designed a deterministic algorithm, we would like to allow for practitioners to incorporate exploration in choosing $x$ by allowing $a_h(x)$ to be a function of a random variable.
This would mean that $\mathcal{A}_C$ should be \textit{seeded} in practice so that it outputs the same quantities for the same random seed.
Indeed, we implement this technique in our experiments.
Along the same vein, since $\hat{e}_{h,C}$ should induce a total order, then if there are ties, where $\hat{e}_{h, C}(x_i) = \hat{e}_{h, C}(x_j)$ for $x_i \neq x_j$, a tie-breaking criterion can be implemented using a small amount of random noise: $\hat{e}_{h, C}(x) := a_{h}(x) + \lambda \mu_{h}(x) + \sigma(x)$ for positive-valued $\sigma(x)$ such that $\hat{e}_{h, C}(x) \geq a_{h}(x) + \lambda \mu_{h}(x) \; \forall x \in \mathcal{X}$ still holds. 
Note that for our algorithm, we only impose that $\hat{e}_{h, C}$ induces a total order and not that $e_h$ induces a total order, in contrast with Pick-to-Learn.
\end{remark}

\section{SIMULATION EXPERIMENTS}\label{sec:simulation_exp}

\noindent\textbf{Experimental Setup.} We consider the example in \cite{li2025certifiable} of an ego drone competing with another drone to reach a gate.
For this reachability problem, the state is $12$-dimensional, consisting of the two agents' $6$-dimensional dynamics that are described by double integrators in the three spatial dimensions.
The reward function $r(x)$ incentivizes the ego drone to remain ahead of the other drone, move faster than the other drone, and approach the gate within a corridor.
The constraint function $c(x)$ imposes that the ego drone approach the gate within a certain angle, maintain an altitude limit, and stay away from the other agent and its approximate region of downwash. 
We produce $V_{\tilde{\pi}}(x,T)$ values by applying $\tilde{\pi}$ for the ego drone and applying a PID controller for the other drone for $T=30$ steps.

In order to avoid retraining a neural network at every iteration, we model $h_0(x)$ as a GP that uses $\tilde{V}(x,t)$ as part of its prior.
Thus, the total sample complexity of our method for these experiments is $|Q| + |C|$ plus 40 initial samples for fitting $h_0(x)$.
We opt for a heuristic active learning strategy in which we prioritize points near the boundary of the reachable set, where it is most likely that the model is incorrect.
Since the zero-level set of the value function defines the boundary, $a_{h}(x)$ simply uses the absolute value of the value function and is normalized so that $a_{h}(x) \in [0, 1]$: $a_{h}(x) := 1-|h(x)|/(\max_{x} |h(x)|)$.

\begin{remark}
We are able to keep a low sample complexity by sampling $C$ once and using it throughout the algorithm.
The dependency between $C$ and $h$ is quite weak since $h$ is learned based on only samples from $D$.
Our empirical findings validate the theoretical results of Theorem 1.
\end{remark}

Note that our algorithm terminates when $\hat{e}_{h,C}(x) \leq \omega  \; \forall x \in D \setminus Q$, and since $\hat{e}_{h, C}(x) \geq a_{h}(x)$, if $a_{h}(x)$ is never below $\omega \; \forall x \in D \setminus Q$ during the course of the algorithm then neither is $\hat{e}_{h,C}(x)$.
This is indeed the case for our particular strategy $a_{h}(x)$.
To remedy this, we adopt the engineering solution of decaying $a_{h}(x)$ at every iteration of the while loop in Algorithm \ref{alg:approx_picktolearn} before the conformal prediction step is applied using a multiplicative exponential factor $\zeta \in [0,1]$, thus maintaining the guarantee.
Also, we heuristically define $\mu_{h}(x)$ as $0.1 \times a_{h}(x)$.
The value of $\hat{e}_{h,C}(x)$ at every iteration depends on the values of the score functions in the calibration dataset.
If for some $z_i^C \in C$ $a_{h}(z_{i,x}^C)$ is small but $e_h(z_i^C) = 1$, $\hat{e}_{h,C}(x)$ may remain above $\omega$ for some $x \in D$.
It is only when discrepancies between $a_{h}$ and $e_h$ lessen, usually by values of $e_h$ in $C$ going to zero, that the algorithm is allowed to terminate.
We set $\omega = 0.3$, $\zeta = 0.95$, and $\delta = 1e^{-4}$.
\\

\begin{figure}[!t]
\centering
\includegraphics[width=2.7in]{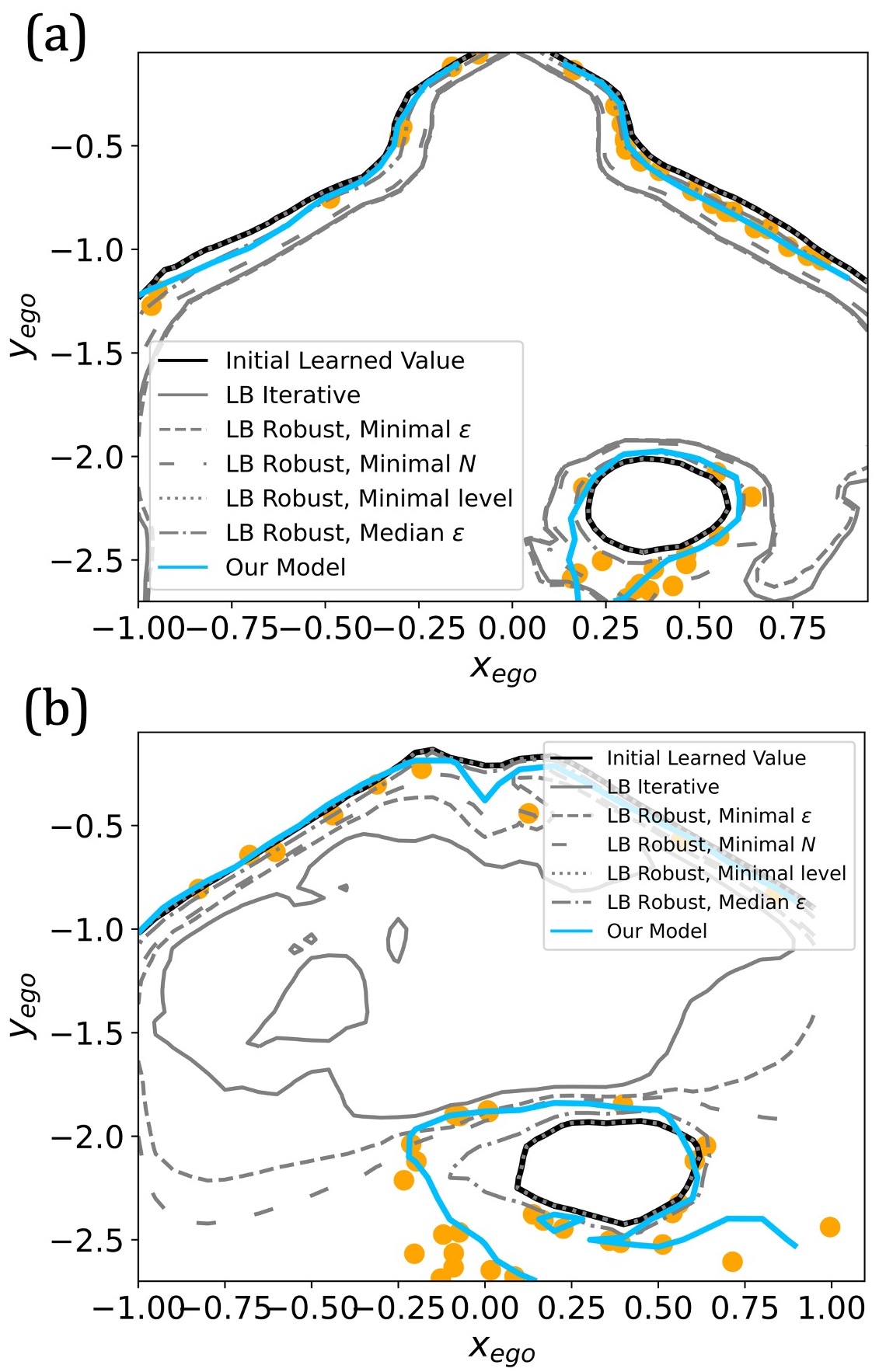}\vspace{-0.4em} 
\caption{(a) In this 2-dimensional ``Slice 1'', the ego drone's 3D velocity is set to $[0.0, 0.7, 0.0]$ and its altitude is $0.0$.
The other drone's 3D spatial coordinates are $[0.4, -2.2, 0.0]$ and its 3D velocity is $[0.0, 0.3, 0.0]$.
(b) In this 2-dimensional ``Slice 2'', the other drone's state is the same, but the ego drone's velocity is $[0.0, 0.0, -0.5]$ and its altitude is $0.05$.
The baselines' calibration technique is to simply choose an appropriate level of the learned value function.
For the more commonly seen scenario (a), the level sets have more regular shapes, better justifying this calibration technique.
However, for the less commonly seen scenario (b), this is not the case, and it is more reasonable to be unconstrained by learned set geometries, as in our method.
In both figures, orange points are the samples that our method took.
}
\label{fig:slices}
\end{figure}

\noindent\textbf{Baselines and Metrics.} We compare our method with the two baselines described in Section \ref{sec:related_work_A}: \cite{lin2023generating}, which we call LB Iterative after the names of the authors, and \cite{lin2024verification}, which we denote as LB Robust. 
Both guarantee $P_{x \in \hat{\mathcal{S}}}(V_{\tilde{\pi}}(x,T) \leq 0) \leq \epsilon_{LB}$ with probability at least $1-\beta$ over sampled data, where $\hat{\mathcal{S}}$ is the calibrated set. We set $\beta = 0.1$ for all methods. 
Notably, our probabilistic guarantee is with respect to a general notion of error, whereas the baselines' guarantee only constrains the false inclusion of unsafe states.

We implement the LB Robust method to highlight the trade-offs involved in calibration.
We vary the number of samples over ($[50, 200, 250, 300, 350, 500, 750, 1000]$) and the levels of the learned value function ($[0.0, 0.05, 0.1, 0.15,$ $ 0.2, 0.3, 0.5, 0.75, 0.9, 1.0]$).
For each combination of these, we implement the method to obtain the $\epsilon_{LB}$ value.
Then, over all combinations we take the one with the minimal $\epsilon_{LB}$ (``LB Robust, Minimal $\epsilon$''), the minimal $\epsilon_{LB}$ for the smallest number of samples (``LB Robust, Minimal $N$''), the minimal $\epsilon_{LB}$ for the smallest level (``LB Robust, Minimal level''), and finally the median $\epsilon_{LB}$ (``LB Robust, Median $\epsilon$'').

For all methods, we compare the total number of samples required, $\epsilon_{LB}$ and $\bar{\epsilon}(|Q|, \delta) + \alpha + \epsilon_{\alpha}$, and the false positive rate and false negative rate (FPR and FNR) of the final sets.
Here, FPR and FNR are calculated by gridding the state space, computing the ground truth values at the grid points, and checking whether the calibrated set includes (a ``positive'') and excludes (a ``negative'') the correct points.
\\

\noindent\textbf{Results and Discussion.} We vary the type of active learning technique we use, the dimension and slice of the state space that we operate in, and the values of $\alpha$ and $\epsilon_{\alpha}$. 
Fig. \ref{fig:slices} visualizes the impact of some of these variations in 2-dimensional slices of the 12-dimensional system.
Compared to Fig. \ref{fig:slices}(a), the states in the slice in Fig. \ref{fig:slices}(b) are less likely to have been seen during the learning of $\tilde{V}(x,t)$.
Since the baseline methods can only choose among the levels of the learned value function, they struggle more when the learned function is less accurate.
Our method, meanwhile, is not constrained by the geometries of the learned level sets.

\begin{figure*}[t]
\centering
\includegraphics[width=7.0in]{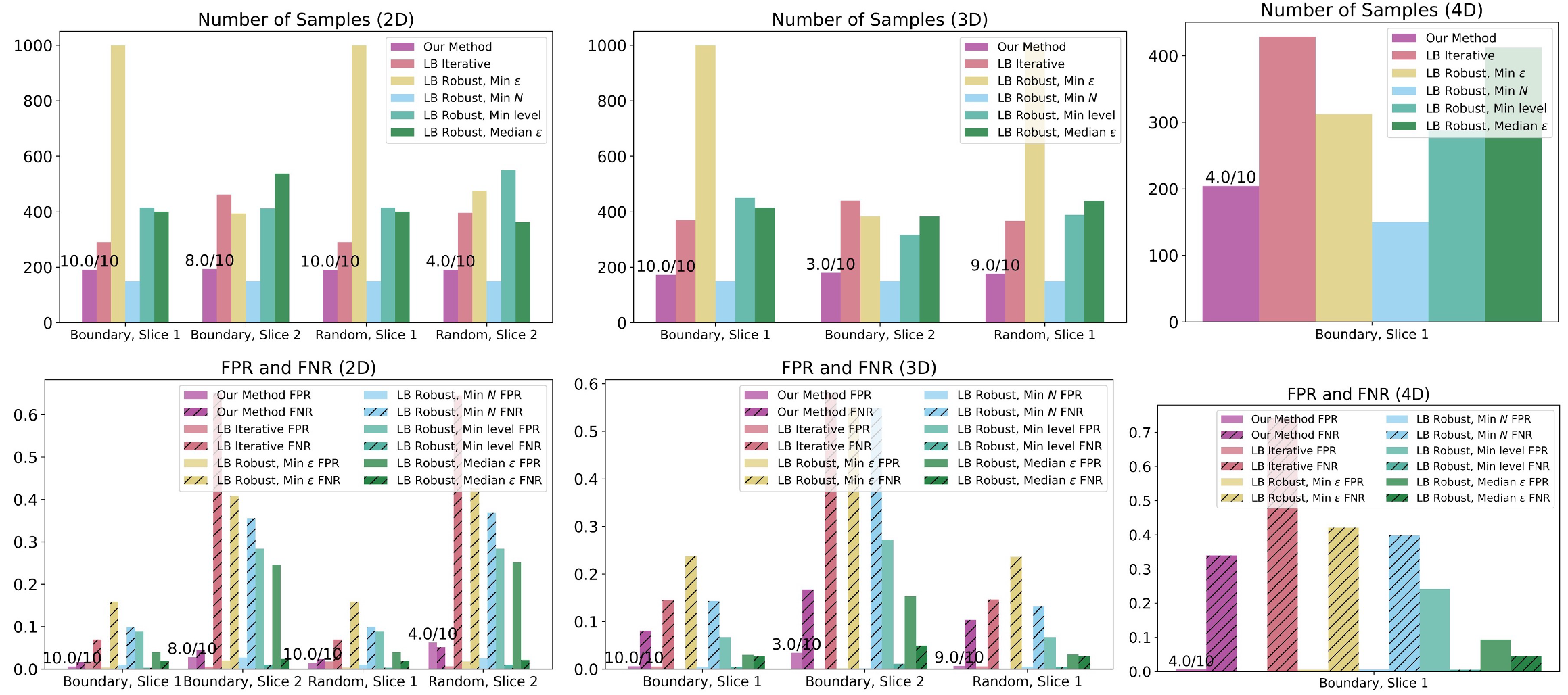}\vspace{-0.2em}
\caption{Plots comparing the number of samples and FPR/FNR across methods.
``Boundary'' denotes our boundary sampling active learning technique, while ``Random'' refers to setting $a_{h}(x)$ to random values.
In 3D experiments, the slice of the dynamics is the same as that for the 2D Slice 1 except the ego drone's altitude is also allowed to vary.
In 4D experiments, we additionally allow the ego agent's velocity in the x-direction to vary.
The numbers above the bars indicate how many seeds out of 10 were successful for our method.
Our method balances the competing objectives of minimizing the number of samples and FPR and FNR well compared with the baselines.
As we expect, performance for all methods is worse for Slice 2 compared with Slice 1.
}
\label{fig:bar_charts}
\end{figure*}

Fig. \ref{fig:bar_charts} provides comparisons of the methods for different variable combinations.
Each experiment is run for 10 random seeds, and comparisons are only made for successful seeds.
Here, failure is defined as more than 70 iterations of our algorithm, since the method is then no longer adaptive with minimal samples.
Our method uses the least number of samples except for the ``LB Robust, Minimal $N$'' method, which has a worse FNR. 
As seen in Table \ref{tab:N_results}, the tightest probabilistic bound is achieved by the ``LB Robust, Minimal $\epsilon$'' method, although it also uses the most samples.
Combining conformal prediction with Pick-to-Learn in our work weakens our guarantee in comparison.
However, we emphasize that this is because our method incorporates \textit{learning and adaptation} with respect to any error function, which is necessary for many real-world scenarios.
Our method also more effectively trades off between FNR and FPR compared with the baselines, which may minimize one at the expense of the other.

Across our method variants, boundary and random active learning require a similar number of samples. Though seemingly counterintuitive, this reflects a key property of probabilistic guarantees.
The required number of iterations depends on the particular calibration set $C$ and how well the model learns per sample.
Especially if $\alpha$ and $\epsilon_{\alpha}$ are larger values, the $\hat{e}_{h,C}(x) \leq \omega$ condition may not be as difficult to satisfy within a set number of iterations.
However, even if the Boundary and Random experiments yield a similar probabilistic guarantee due to using a similar number of samples, the final reachable set learned via the boundary sampling strategy tends to have a better FPR and FNR.

Finally, we study how our method's performance scales with increasing state dimension, where model learning becomes more difficult and reducing errors in $C$ becomes harder. In these cases, relaxing $\alpha$ and $\epsilon_{\alpha}$ reduces the required number of samples, at the cost of higher FPR and FNR. 
We use $(\alpha,\epsilon_\alpha)=(0.05, 0.03)$ for the 2D experiments ($|C| = 112$, $|D| = 4000$), $(0.1, 0.05)$ for the 3D experiments ($|C| = 93$, $|D| = 4000$), and $(0.15, 0.05$) for the 4D experiments ($|C| = 112$, $|D| = 8000$). 
By uniting the processes of learning
and probabilistic assurance, our framework allows for a formal examination
of how they interact and trade off.

\begin{table}[t]

\caption{
Comparing $\bar{\epsilon}(|Q|, \delta) + \alpha + \epsilon_{\alpha}$ and $\epsilon_{LB}$ values to understand the tightness of the probabilistic bounds - lower values generate stronger statements. ``LB Robust, Min $\epsilon$'' always has the lowest $\epsilon$, but it uses the most samples.
} \label{tab:N_results}

\vspace{0.3em}
\centering 
\setlength{\tabcolsep}{4.1pt} 
\begin{tabular}{|c|*{6}{c|}}

\hline

Setting & Ours & LB & LB & LB & LB & LB \\

\hfill & \hfill & Iterative & Robust & Robust & Robust & Robust  \\

& \hfill & \hfill & Min $\epsilon$ & Min $N$ & Min Level & Median $\epsilon$ \\

\hline

\multicolumn{3}{|c}{} & \multicolumn{1}{c}{2D} & \multicolumn{3}{c|}{}

\\

\hline 

Boundary, & 0.099 & 0.05 & 0.003 & 0.016 & 0.035 &  0.024
 \\

Slice 1 & \hfill & \hfill & \hfill & \hfill & \hfill
 & \hfill \\

\hline

Boundary, & 0.1 & 0.05 & 0.01 & 0.021 & 0.093 & 0.094 \\

Slice 2 & \hfill & \hfill & \hfill & \hfill & \hfill
 & \hfill \\

\hline

Random, & 0.099 & 0.05 & 0.003 & 0.016 & 0.035 & 0.024 \\

Slice 1 & \hfill & \hfill & \hfill & \hfill & \hfill
 & \hfill \\

\hline

Random, & 0.099 & 0.05 & 0.009 & 0.019 & 0.096 & 0.095 \\

Slice 2 & \hfill & \hfill & \hfill & \hfill & \hfill
 & \hfill \\

\hline

\multicolumn{3}{|c}{} & \multicolumn{1}{c}{3D} & \multicolumn{3}{c|}{}

\\

\hline 

Boundary, & 0.169 & 0.05 & 0.003 & 0.016 & 0.040 & 0.028

 \\

Slice 1 & \hfill & \hfill & \hfill & \hfill & \hfill
 & \hfill \\

\hline

Boundary, & 0.172 & 0.05 & 0.007 &  0.016 & 0.160 & 0.118 \\

Slice 2 & \hfill & \hfill & \hfill & \hfill & \hfill
 & \hfill \\

\hline

Random, & 0.171 & 0.05 & 0.003 &  0.016 & 0.039 & 0.028 \\

Slice 1 & \hfill & \hfill & \hfill & \hfill & \hfill
 & \hfill \\

\hline

\multicolumn{3}{|c}{} & \multicolumn{1}{c}{4D} & \multicolumn{3}{c|}{}

\\

\hline 

Boundary, & 0.212 & 0.05 & 0.015 & 0.026 & 0.203 & 0.111 \\

Slice 1 & \hfill & \hfill & \hfill & \hfill & \hfill
 & \hfill \\

\hline

\end{tabular}

\end{table}

\section{CONCLUSIONS}\label{sec:conclusion}

In this work, we (1) introduce a method to use active learning to both use fewer samples and gain more accuracy when calibrating reachable sets and (2) produce an extension to the Pick-to-Learn algorithm that allows us to provide a probabilistic guarantee for active learning.
In this way, our work allows for an adaptive notion of safety.
We believe that our method can be used not only for calibration but also for reachable set synthesis more generally.
In our next steps, we would like to apply these principles in hardware experiments.
From our analyses, we additionally see that our intuitive understanding of risk in dynamical systems does not always align with the probabilistic statements made by statistical techniques. 
For example, the size of $C$ prescribed by conformal prediction theory is not dependent on the dimension of the state space even though it is intended to capture model error.
Also, the relationship between probabilistic guarantees, set conservatism, and learning difficulty warrants further investigation in future work.

\bibliographystyle{ieeetr}
\bibliography{reference}

\section*{APPENDIX}

In this Appendix, we provide a full proof of Theorem 1.
In order to do so, we introduce several Lemmas.
All theoretical statements are provided below.

\subsection{Theoretical Statements}
\begin{lemma}
    Let $\kappa_{C}$ be our compression function; that is, $\kappa_{C}(D) := Q$.
    Then, $\kappa_{C}$ is a preferent compression function given $C$.
\end{lemma}\vspace{0.4em}

\begin{lemma}
    For any $h$ present during the execution of $\mathcal{A}_{C}$, \\ $P_C[P_z(e_h(z) \geq \omega \land \hat{e}_{h,C}(z_x) \leq \omega  \; | \; C,D) \leq \alpha + \epsilon_{\alpha} \; | \; D] \geq 1-\beta$ for the threshold $\omega$.
\end{lemma}\vspace{0.4em}

\begin{lemma}
$e_{h_f}(z) \geq \omega \; \land \; \kappa_{C}(\kappa_{C}(D) \cup \{z\}) = \kappa_{C}(D) \implies e_{h_f}(z) \geq \omega \; \land \; \hat{e}_{h_f,C}(z_x) \leq \omega$.
\end{lemma}\vspace{0.4em}

\begin{lemma}
Given logical statements $A(C,D)$ and $B(C,D)$ for random variables $C$ and $D$, if $P_C(A(C,D) | D) \geq 1-\beta$ and $P_D(B(C,D) | C) \geq 1-\delta$, then $P_{C,D}(A(C,D) \land B(C,D)) \geq 1-\beta-\delta$.
\end{lemma}\vspace{0.4em}

\noindent \textbf{Theorem 1. }
Let $\mathcal{A}_{C}$ be Algorithm 1, such that $\mathcal{A}_{C}(D) = (h_f, Q)$. Pick $\alpha, \epsilon_{\alpha}, \beta,\delta \in (0, 1)$ and calculate $\bar{\epsilon}(|Q|, \delta)$ as in Pick-to-Learn.
Then,
\begin{equation*}
    P_{C,D}[P_z(e_{h_f}(z) \geq \omega \; | \; C,D) \leq \bar{\epsilon}(|Q|, \delta) + \alpha + \epsilon_{\alpha}] \geq 1-\beta-\delta.
\end{equation*}

\subsection{Proof of Lemma 1}

\textbf{Lemma 1.} Let $\kappa_{C}$ be our compression function; that is, $\kappa_{C}(D) := Q$.
Then, $\kappa_{C}$ is a preferent compression function given $C$.
\\

\begin{proof}
If $C$ is given, the proof of this statement follows the proof of Lemma A.6 from Paccagnan et al. (2023) \cite{paccagnan2023picktolearn} since $\hat{e}_{h,C}$ defines a total order.
The proof shows that, given that $\kappa_{C}(D \cup \{z\}) \neq \kappa_{C}(D)$, that  $\kappa_{C}(D \cup \{z\}) \neq V$ for any $V \subseteq D$. 
\end{proof}

\subsection{Proof of Lemma 2}

\textbf{Lemma 2. } For any $h$ present during the execution of $\mathcal{A}_{C}$, \\ $P_C[P_z(e_h(z) \geq \omega \land \hat{e}_{h,C}(z_x) \leq \omega  \; | \; C,D) \leq \alpha + \epsilon_{\alpha} \; | \; D] \geq 1-\beta$ for the threshold $\omega$.

\begin{proof}
For all $h$ present during the execution of $\mathcal{A}_{C}$, our methodology guarantees, via conditional conformal prediction, that

\begin{align*}
    P_C[1-\alpha-\epsilon_{\alpha} \leq P_z(e_h(z) \in [a_{h}(z_x) - \lambda_{h,C}\mu_{h,\eta}(z_x), \\
    a_{h}(z_x) + \lambda_{h,C}\mu_{h}(z_x)] \; | \; C,D) \leq 1-\alpha+\epsilon_{\alpha} \; | \; D] \geq 1-\beta
\end{align*}

for suitable values of $\alpha, \epsilon_{\alpha}, \beta$, and $|C|$, for $\lambda_{h,C}$ computed via conformal prediction, for any $z$ selected independently from $C$ and $D$, and for any arbitrary $a_{h}(x)$ and $\mu_{h}(x)$.
We impose that $\mu_{h}(x)$ must never be zero.
We write $\hat{e}_{h,C}(x) := a_{h}(x) + \lambda_{h,C}\mu_{h}(x)$.

Rewriting the guarantee from conformal prediction, we have

\begin{center}
    $P_C[1-\alpha-\epsilon_{\alpha} \leq P_z(a_{h}(z_x) - \lambda_{h,C}\mu_{h}(z_x) \leq e_h(z) \leq a_{h}(z_x) + \lambda_{h,C}\mu_{h}(z_x) \; | \; C,D) \leq 1-\alpha+\epsilon_{\alpha} \; | \; D] \geq 1-\beta$
\end{center}

Since $1-\alpha-\epsilon_{\alpha} \leq P_z(a_{h}(z_x) - \lambda_{h,C}\mu_{h}(z_x) \leq e_h(z) \leq a_{h}(z_x) + \lambda_{h,C}\mu_{h}(z_x) \; | \; C,D) \leq 1-\alpha+\epsilon_{\alpha}$ implies $1-\alpha-\epsilon_{\alpha} \leq P_z(a_{h}(z_x) - \lambda_{h,C}\mu_{h}(z_x) \leq e_h(z) \leq a_{h}(z_x) + \lambda_{h,C}\mu_{h}(z_x) \; | \; C,D)$, we have

\begin{center}
    $P_C[1-\alpha-\epsilon_{\alpha} \leq P_z(a_{h}(z_x) - \lambda_{h,C}\mu_{h}(z_x) \leq e_h(z) \leq a_{h}(z_x) + \lambda_{h,C}\mu_{h}(z_x) \; | \; C,D) \; | \; D] \geq 1-\beta$
\end{center}

Since $a_{h}(z_x) - \lambda_{h,C}\mu_{h}(z_x) \leq e_h(z) \leq a_{h}(z_x) + \lambda_{h,C}\mu_{h}(z_x)$ implies $e_h(z) \leq a_{h}(z_x) + \lambda_{h,C}\mu_{h}(z_x)$, 

\begin{center}
    $P_z(e_h(z) \leq a_{h}(z_x) + \lambda_{h,C}\mu_{h}(z_x) \; | \; C,D) \geq$ \\
    $P_z(a_{h}(z_x) - \lambda_{h,C}\mu_{h}(z_x) \leq e_h(z) \leq a_{h}(z_x) + \lambda_{h,C}\mu_{h}(z_x) \; | \; C,D)$
\end{center}

is always true.
So,

\begin{center}
    $P_C[1-\alpha-\epsilon_{\alpha} \leq P_z(e_h(z) \leq a_{h}(z_x) + \lambda_{h,C}\mu_{h}(z_x) \; | \; C,D)\; | \; D ] \geq 1-\beta$ 
\end{center}

and substituting $\hat{e}_{h,C}$,

\begin{center}
    $P_C[1-\alpha-\epsilon_{\alpha} \leq P_z(e_h(z) \leq \hat{e}_{h,C}(z_x) \; | \; C,D) \; | \; D] \geq 1-\beta$
\end{center}

This means that

\begin{center}
    $P_C[P_z(e_h(z) \geq \hat{e}_{h,C}(z_x) \; | \; C,D) \leq \alpha+\epsilon_{\alpha} \; | \; D] \geq 1-\beta$
\end{center}

Now we use the Law of Total Probability for the inner probability $P_z(e_h(z) \geq \hat{e}_{h,C}(z_x) \; | \; C,D)$:

\begin{center}
    $P_z(e_h(z) \geq \hat{e}_{h,C}(z_x) \; | \; C,D) = $ \\
    $P_z(e_h(z) \geq \omega \land e_h(z) \geq \hat{e}_{h,C}(z_x) \geq \omega  \; | \; C,D) + P_z(e_h(z) \geq \omega \land \hat{e}_{h,C}(z_x) \leq \omega  \; | \; C,D) +$ \\
    $P_z(e_h(z) \leq \omega \land \hat{e}_{h,C}(z_x) \leq e_h(z) \leq \omega  \; | \; C,D)$
\end{center}

since $\omega$ is a fixed deterministic value. 

So, $P_z(e_h(z) \geq \hat{e}_{h,C}(z_x) \; | \; C,D) \leq \alpha+\epsilon_{\alpha}$ implies that $P_z(e_h(z) \geq \omega \land \hat{e}_{h,C}(z_x) \leq \omega  \; | \; C,D) \leq \alpha + \epsilon_{\alpha}$.
This gives

\begin{center}
    $P_C[P_z(e_h(z) \geq \omega \land \hat{e}_{h,C}(z_x) \leq \omega  \; | \; C,D) \leq \alpha + \epsilon_{\alpha}] \geq P_C[P_z(e_h(z) \geq \hat{e}_{h,C}(z_x) \; | \; C,D) \leq \alpha+\epsilon_{\alpha} \; | \; D] \geq 1-\beta$
\end{center}

and 

\begin{center}
    $P_C[P_z(e_h(z) \geq \omega \land \hat{e}_{h,C}(z_x) \leq \omega  \; | \; C,D) \leq \alpha + \epsilon_{\alpha} \; | \; D] \geq 1-\beta$.
\end{center}
\end{proof}

\subsection{Proof of Lemma 3}

\textbf{Lemma 3. } $e_{h_f}(z) \geq \omega \; \land \; \kappa_{C}(\kappa_{C}(D) \cup \{z\}) = \kappa_{C}(D) \implies e_{h_f}(z) \geq \omega \; \land \; \hat{e}_{h_f,C}(z_x) \leq \omega$.

\begin{proof}
We will show that $\kappa_{C}(\kappa_{C}(D) \cup \{z\}) = \kappa_{C}(D) \implies \hat{e}_{h_f,C}(z_x) \leq \omega$.
Then, the statement follows. \\

If $z$ is such that $\kappa_{C}(\kappa_{C}(D) \cup \{z\}) = \kappa_{C}(D)$, or $\kappa_{C}(Q \cup \{z\}) = Q$, then $\mathcal{A}_{C}(Q \cup \{z\})$ must select all of the same points as $\mathcal{A}_{C}(D)$, choosing the same points at every iteration.
It must terminate having produced $Q$ and $h_f$.
This must mean that in its last iteration, $\arg \max_{\tilde{z} \in Q \cup \{z\} \setminus Q} \hat{e}_{h_f,C}(\tilde{z}_x) \leq \omega$, which implies that $\arg \max_{\tilde{z} \in \{z\}} \hat{e}_{h_f,C}(\tilde{z}_x) \leq \omega$.
Then, $\hat{e}_{h_f,C}(z_x) \leq \omega$.
\end{proof}

\subsection{Proof of Lemma 4}

\textbf{Lemma 4. } Given logical statements $A(C,D)$ and $B(C,D)$ for random variables $C$ and $D$, if $P_C(A(C,D) | D) \geq 1-\beta$ and $P_D(B(C,D) | C) \geq 1-\delta$, then $P_{C,D}(A(C,D) \land B(C,D)) \geq 1-\beta-\delta$.

\begin{proof}
We have

\begin{center}
    $P_C(A(C,D)|D) \geq 1-\beta \implies$ \\
    $P_C(\neg A(C,D)|D) \leq \beta$
\end{center}

Since

\begin{center}
    $P_{C,D}(\neg A(C,D)) = \mathbb{E}_D[P_C(\neg A(C,D)|D)]$
\end{center}

and $P_C(\neg A(C,D)|D) \leq \beta$, it follows that $\mathbb{E}_D[P_C(\neg A(C,D)|D)] \leq \mathbb{E}[\beta] = \beta$. 
So, $P_{C,D}(\neg A(C,D)) \leq \beta$.
The reasoning is similar for $P_{C,D}(\neg B(C,D)) \leq \delta$. \\

Now using the union bound,

\begin{center}
    $P_{C,D}(\neg A(C,D) \lor \neg B(C,D)) \leq \beta + \delta$
\end{center}

and 

\begin{center}
    $P_{C,D}(A(C,D) \land B(C,D)) \geq 1-\beta-\delta$.
\end{center}
\end{proof}

\subsection{Proof of Theorem 1}

\textbf{Theorem 1.} Let $\mathcal{A}_{C}$ be Algorithm 1, such that $\mathcal{A}_{C}(D) = (h_f, Q)$. Pick $\alpha, \epsilon_{\alpha}, \beta,\delta \in (0, 1)$ and calculate $\bar{\epsilon}(|Q|, \delta)$ as in Pick-to-Learn~\cite{paccagnan2023picktolearn}.
Then,
\begin{equation*}
    P_{C,D}[P_z(e_{h_f}(z) \geq \omega \; | \; C,D) \leq \bar{\epsilon}(|Q|, \delta) + \alpha + \epsilon_{\alpha}] \geq 1-\beta-\delta.
\end{equation*}

\begin{proof}
\textbf{Lemma 2} implies that $P_C[P_z(e_{h_f}(z) \geq \omega \land \hat{e}_{h_f,C}(z_x) \leq \omega  \; | \; C,D) \leq \alpha + \epsilon_{\alpha} \; | \; D] \geq 1-\beta$. \textbf{Lemma 3} implies that

\begin{center}
    $P_z(e_{h_f}(z) \geq \omega \land \kappa_{C}(\kappa_{C}(D) \cup \{z\}) = \kappa_{C}(D))  \; | \; C,D) \leq P_z(e_{h_f}(z) \geq \omega \land \hat{e}_{h_f,C}(z_x) \leq \omega  \; | \; C,D)$
\end{center}

is always true. So,

\begin{center}
    $P_C[P_z(e_{h_f}(z) \geq \omega \land \kappa_{C}(\kappa_{C}(D) \cup \{z\}) = \kappa_{C}(D))  \; | \; C,D) \leq \alpha + \epsilon_{\alpha} \; | \; D] \geq 1-\beta$
\end{center}

Now, we use the Law of Total Probability on the inner probability to write:

\begin{center}
    $P_z(e_{h_f}(z) \geq \omega \; | \; C,D) = P_z(e_{h_f}(z) \geq \omega \land \kappa_{C}(\kappa_{C}(D) \cup \{z\}) = \kappa_{C}(D) \; | \; C,D) + $\\
    $P_z(e_{h_f}(z) \geq \omega \land \kappa_{C}(\kappa_{C}(D) \cup \{z\}) \neq \kappa_{C}(D) \; | \; C,D)$
\end{center}

We know that 
\begin{center}
    $P_z(e_{h_f}(z) \geq \omega \land \kappa_{C}(\kappa_{C}(D) \cup \{z\}) \neq \kappa_{C}(D) \; | \; C,D) \leq $ \\
    $P_z(\kappa_{C}(\kappa_{C}(D) \cup \{z\}) \neq \kappa_{C}(D) \; | \; C,D)$
\end{center}
by the monotonicity of probabilities.
So, $P_z(e_{h_f}(z) \geq \omega \land \kappa_{C}(\kappa_{C}(D) \cup \{z\}) = \kappa_{C}(D))  \; | \; C,D) \leq \alpha + \epsilon_{\alpha}$ implies $P_z(e_{h_f}(z) \geq \omega \; | \; C,D) \leq P_z(\kappa_{C}(\kappa_{C}(D) \cup \{z\}) \neq \kappa_{C}(D) \; | \; C,D) + \alpha + \epsilon_{\alpha}$ always.
So we can write

\begin{center}
    $P_C[P_z(e_{h_f}(z) \geq \omega \; | \; C,D) \leq P_z(\kappa_{C}(\kappa_{C}(D) \cup \{z\}) \neq \kappa_{C}(D) \; | \; C,D) + \alpha + \epsilon_{\alpha} \; | \; D] \geq$ \\
    $P_C[P_z(e_{h_f}(z) \geq \omega \land \kappa_{C}(\kappa_{C}(D) \cup \{z\}) = \kappa_{C}(D))  \; | \; C,D) \leq \alpha + \epsilon_{\alpha} \; | \; D] \geq 1-\beta$
\end{center}

or simply

\begin{center}
    $P_C[P_z(e_{h_f}(z) \geq \omega \; | \; C,D) \leq P_z(\kappa_{C}(\kappa_{C}(D) \cup \{z\}) \neq \kappa_{C}(D) \; | \; C,D) + \alpha + \epsilon_{\alpha} \; | \; D] \geq 1-\beta$.
\end{center}

Theorem 7 of~\cite{campi2023compression} states that, if $\kappa_{C}$ satisfies the property of preference, then

\begin{center}
$P_D[P_z(\kappa_{C}(\kappa_{C}(D) \cup \{z\}) \neq \kappa_{C}(D) | D) \leq \bar{\epsilon}(|Q|, \delta)] \geq 1-\delta$
\end{center}

We claim that this is equivalent to 

\begin{center}
$P_D[P_z(\kappa_{C}(\kappa_{C}(D) \cup \{z\}) \neq \kappa_{C}(D) | D, C) \leq \bar{\epsilon}(|Q|, \delta) \; | \; C] \geq 1-\delta$
\end{center}

and will explain this reasoning at the end of this proof.
We invoke \textbf{Lemma 1} (our compression function is preferent) to use this theorem. \\

For clarity, let us shorthand the statement $P_z(e_{h_f}(z) \geq \omega \; | \; C,D) \leq P_z(\kappa_{C}(\kappa_{C}(D) \cup \{z\}) \neq \kappa_{C}(D) \; | \; C,D) + \alpha + \epsilon_{\alpha}$ as $A(C, D)$ and the statement $P_z(\kappa_{C}(\kappa_{C}(D) \cup \{z\}) \neq \kappa_{C}(D) | D, C) \leq \bar{\epsilon}(|Q|, \delta)$ as  $B(C, D)$.
So, our logic so far has produced the statements

\begin{center}
    $P_C(A(C,D)|D) \geq 1-\beta$ \\
    $P_D(B(C,D)|C) \geq 1-\delta$
\end{center}

$A(C, D) \land B(C, D)$ is equal to

\begin{center}
    $P_z(e_{h_f}(z) \geq \omega \; | \; C,D) \leq \bar{\epsilon}(|Q|, \delta) + \alpha + \epsilon_{\alpha}$
\end{center}

$A(C,D)$ is true with a probability over $C$ (though for all $D$), and $B(C,D)$ is true with a probability over $D$ (though for all $C$).
Thus, we can use \textbf{Lemma 4} to merge the two statements and give our final statement:

\begin{center}
    $P_{C,D}[P_z(e_{h_f}(z) \geq \omega \; | \; C,D) \leq \epsilon(|Q|, \delta) + \alpha + \bar{\epsilon}_{\alpha}] \geq 1-\beta-\delta$.
\end{center}

Now it remains to explain why we can write Theorem 7 in~\cite{campi2023compression} as 

\begin{center}
$P_D[P_z(\kappa_{C}(\kappa_{C}(D) \cup \{z\}) \neq \kappa_{C}(D) | D, C) \leq \bar{\epsilon}(|Q|, \delta) \; | \; C] \geq 1-\delta$.
\end{center}

We argue that the original theorem did not consider a compression function that is a function of any random variable other than $D$, while in our case $\kappa$ is also dependent on $C$.
This means that in our case, the probability of change of compression is a function of $D$ and $C$ and not just $D$.
Therefore, the probability of change of compression should be written as $P_z(\kappa_{C}(\kappa_{C}(D) \cup \{z\}) \neq \kappa_{C}(D) | D, C)$ for correctness.
The reasoning in the proof of Theorem 4 in~\cite{campi2023compression} about the distribution of the probability of change of compression with respect to $D$ does not change.
One can think of our proof as adding reasoning on the distribution of the probability of change of compression with respect to $C$ \textit{in addition to} $D$.

\end{proof}

\end{document}